\documentclass{article}
\usepackage{hyperref}

\usepackage[nompar,nosign]{commenting}

\title{$\exists$\GN\ for Boolean Games is NEXP-hard}
\author{Egor Ianovski \and Luke Ong\\
Department of Computer Science, University of Oxford\\
Wolfson Building, Parks Road, Oxford, UK\\}

\usepackage{amsmath, amsthm, amssymb}

\newtheorem{theorem}{Theorem}
\newtheorem{proposition}{Proposition}
\newtheorem{lemma}{Lemma}
\newtheorem{corollary}{Corollary}

\theoremstyle{definition}
\newtheorem{definition}{Definition}
\newtheorem{example}{Example}

\newcommand\st{\, . \,}

\newcommand\Zero{\mathit{Zero}}
\newcommand\pZero{\mathit{pZero}}
\newcommand\sZero{\mathit{sZero}}
\newcommand\nZero{\mathit{nZero}}
\newcommand\npZero{\mathit{npZero}}
\newcommand\nsZero{\mathit{nsZero}}
\newcommand\One{\mathit{One}}
\newcommand\pOne{\mathit{pOne}}
\newcommand\sOne{\mathit{sOne}}
\newcommand\nOne{\mathit{nOne}}
\newcommand\npOne{\mathit{npOne}}
\newcommand\nsOne{\mathit{nsOne}}
\newcommand\Head{\mathit{Head}}
\newcommand\pHead{\mathit{pHead}}
\newcommand\nHead{\mathit{nHead}}
\newcommand\sHead{\mathit{sHead}}
\newcommand\npHead{\mathit{npHead}}
\newcommand\NoHead{\mathit{NoHead}}
\newcommand\Left{\mathit{Left}}
\newcommand\sLeft{\mathit{sLeft}}

\newcommand\Right{\mathit{Right}}
\newcommand\sRight{\mathit{sRight}}

\newcommand\Time{\mathit{Time}}
\newcommand\nTime{\mathit{nTime}}
\newcommand\State{\mathit{State}}
\newcommand\pState{\mathit{pState}}

\newcommand\npState{\mathit{npState}}
\renewcommand\Tape{\mathit{Tape}}
\newcommand\sTape{\mathit{sTape}}
\newcommand\nTape{\mathit{nTape}}
\newcommand\pTape{\mathit{pTape}}
\newcommand\npTape{\mathit{npTape}}
\newcommand\nsTape{\mathit{nsTape}}
\newcommand\Init{\mathit{Init}}
\newcommand\Final{\mathit{Final}}
\newcommand\Cons{\mathit{Cons}}
\newcommand\Consequent{\mathit{Consequent}}
\newcommand\InitState{\mathit{InitState}}
\newcommand\InitTape{\mathit{InitTape}}
\newcommand\InitHead{\mathit{InitHead}}
\newcommand\BlankCells{\mathit{BlankCells}}

\newcommand\Succ{\mathit{Succ}}
\newcommand\MatchOne{\mathit{MatchOne}}
\newcommand\AgreeTime{\mathit{AgreeTime}}
\newcommand\AgreeCell{\mathit{AgreeCell}}
\newcommand\sAgreeCell{\mathit{sAgreeCell}}
\newcommand\AgreeHead{\mathit{AgreeHead}}
\newcommand\sAgreeHead{\mathit{sAgreeHead}}
\newcommand\OneOf{{\boldsymbol{\mathit{OneOf}}}}
\newcommand\SUCC{{\boldsymbol{\mathit{Succ}}}}
\renewcommand\Succ{\mathit{Succ}}
\newcommand\Triple{\mathit{Triple}}
\newcommand\Rule{\mathit{Rule}}
\newcommand\Rules{\mathit{Rules}}
\newcommand\Centre{\mathit{Centre}}

\newcommand{\ria}{\rightarrow}
\newcommand{\RR}{\mathbb{R}}

\newcommand{\lra}{\leftrightarrow}

\newcommand{\bs}[1]{\boldsymbol{#1}}

\newcommand{\GN}{{\sc GuaranteeNash}}

\begin{document}

\maketitle
  \begin{abstract}Boolean games are an expressive and natural formalism
    through which to investigate problems of strategic interaction
    in multiagent systems. Although they have been widely studied,
    almost all previous work on Nash equilibria in Boolean games
    has focused on the restricted setting of pure strategies.
    This is a shortcoming as finite games are guaranteed to have
    at least one equilibrium in mixed strategies, but many simple games fail
    to have pure strategy equilibria at all.
    We address this by showing that a natural decision problem about mixed
    equilibria: determining whether
    a Boolean game has a mixed strategy equilibrium that guarantees every
    player a given payoff, is NEXP-hard. Accordingly, the
    $\epsilon$ variety of the problem is NEXP-complete.
    The proof can be adapted to show coNEXP-hardness of a similar question: whether all Nash equilibria
    of a Boolean game guarantee every player at least the given
    payoff.
  \end{abstract}

  \section{Introduction}
  
      A multiagent environment makes strategic considerations
      inevitable. Any attempt to explain the behaviour of a system
      consisting of self-interested agents cannot ignore the fact
      that agents' behaviour may be influenced or completely determined
      by the behaviour of other agents in the system. As the field of game theory
      concerns itself with precisely these issues, its concepts find fertile
      ground in the study of multiagent systems.
      
      A shortcoming of game theoretical techniques is that games,
      being combinatorial objects, are liable to get very large
      very quickly. Any computational application of game theory
      would need alternative representations to the normal and
      extensive forms prominent in the economics literature.
      One such representation, based on propositional logic,
      is the Boolean game.
  
      Boolean games were initially introduced as two player
      games which have an algebra isomorphic to the Lindenbaum algebra
      for propositional logic \cite{Harrenstein2001}. Since then Boolean
      games have garnered interest from the multiagent community as a
      simple yet expressive framework to model strategic interaction.
      This has led to the study of complexity issues involved in reasoning about
      these games. While many questions have been answered, the issue
      of mixed strategies remained open.
      
      In this paper we address this lacuna and present the first
      complexity result about mixed equilibria in the Boolean games
      literature: the NEXP-hardness of $\exists$\GN, which asks whether
      a Boolean game has an equilibrium where each player attains
      at least $\bs{v}[i]$ utility, for some input vector $\bs{v}$.
  
    \subsection{Related Work}
    
      Complexity results for Boolean games were first studied in the
      two player case by \cite{Dunne2004} and the $n$-player
      case by \cite{Bonzon2006}, where among other results the authors showed that
      determining the existence of a pure equilibrium is $\Sigma^p_2$-complete in
      the general case, and can be easier should some restrictions be placed
      on the goal formulae of the players. Further enquiry into tractable
      fragments of Boolean games was carried out by \cite{Dunne2012}.
      
      Cardinal extensions to Boolean games were considered as weighted Boolean formula
      games \cite{Mavronicolas2007} and satisfiability games \cite{Bilo2007}, where
      the authors exploited the connection of a subclass of these games to congestion games
      \cite{Rosenthal1973} to obtain complexity results about both mixed and pure equilibria.
    
      A very similar framework to Boolean games is that of Boolean circuit games \cite{Schoenebeck2012}.
      There players are equipped with a Boolean circuit with $k$ input gates and
      $m$ output gates. The input gates are partitioned among the players, and
      a player's strategy is an assignment of values to the gates under his control.
      The output gates encode a binary representation of the player's utility on
      a given input. The authors explore a number of questions about both mixed and
      pure equilibria, including the complexity of $\exists$\GN.
      
      Note that a Boolean game can be seen as a very specific type of Boolean circuit
      game: the players' circuits are restricted to
      $\text{NC}^1$, and the number of output gates to one. Thus easiness
      results for Boolean circuit games directly transfer to Boolean games, and
      hardness results transfer in the other direction. In particular, this means
      that the NEXP-completeness of $\exists$\GN\ for Boolean circuit games proved by \cite{Schoenebeck2012}
      does not imply the result of this paper.
  
  \section{Preliminaries}
  
    While there are many breeds of games in the literature, we here
    restrict ourselves to what is perhaps the most widely used class:
  
    \begin{definition}
      A \emph{finite strategic game} consists of $n$ players,
      each equipped with a finite set of pure strategies, $S_i$,
      and a utility function $u_i:S_1\times\dots\times S_n\ria\RR$.
      
      An $n$-tuple of strategies is called a \emph{strategy profile}: thus
      a utility function maps strategy profiles to the reals.
    \end{definition}
    \begin{example}
      In a game of \emph{matching pennies} two players are given a coin
      each and may choose to display that coin heads or tails up.
      Player Two seeks to match the move of Player One, while player
      one seeks to avoid that. Hence we have $u_2(HH)=u_2(TT)=1$,
      $u_1(HT)=u_1(TH)=1$, and 0 otherwise.
    \end{example}
    
    Note that to represent a finite strategic game explicitly (the 
    \emph{normal form} of the game) we would need to list the players'
    utility on every possible profile. This would require on the order
    of $n|S_i|^n$ entries, taking $S_i$ to mean the size of the ``typical''
    strategy set. Such a representation is both exponential in the number
    of players and linear in the number of strategies: which in itself
    may be very large.
    
    Ideally we would wish to avoid such a representation. If a game has
    some internal structure, it would be natural to ask if the game
    can be described in a more succinct way. In the case where the
    game can be interpreted as players holding propositional preferences
    over Boolean variables the Boolean game offers precisely that.
    
    \begin{definition}
      A \emph{Boolean game} is a representation of a finite strategic
      game given by $n$ disjoint sets of propositional variables, $\Phi_i$,
      and $n$ formulae of propositional logic, $\gamma_i$.
      
      The intended interpretation is that player $i$ controls the
      variables in $\Phi_i$ in an attempt to satisfy $\gamma_i$,
      which may depend on variables not in player $i$'s control.
      The set of pure strategies of player $i$ is then $2^{\Phi_i}$,
      and his utility function is $\nu\mapsto 1$ if $\nu\vDash\gamma_i$,
      and $\nu\mapsto 0$ otherwise.
    \end{definition}
    \begin{example}
      Matching pennies can be given a Boolean representation by setting 
      $\Phi_1=\{p\}$, $\Phi_2=\{q\}$, $\gamma_1=\neg(p\lra q)$ and
      $\gamma_2=p\lra q$.
    \end{example}
    
    The size of a Boolean game is thus on the order of $n(|\Phi_i|+|\gamma_i|)$.
    In the best case $\gamma_i$ is small and the resulting representation is
    linear in the number of players and logarithmic in the number of strategies,
    giving greater succinctness on both fronts.
    
    Having defined the game representation, we now turn to reasoning about
    such games. The most common solution concept is the Nash equilibrium,
    which we define below.
    
    \begin{definition}
      Given a strategy profile $\bs{s}$, we use $\bs{s}_{-i}(\sigma_i')$ to mean
      the profile obtained by replacing the strategy of $i$ in $\bs{s}$ with 
      $\sigma_i'$. A \emph{best response} for $i$ to $\bs{s}$ is some
      $\sigma_i'$ that maximises $u_i(\bs{s}_{-i}(\sigma_i'))$.
      
      A strategy profile $\bs{s}=(\sigma_1,\dots,\sigma_n)$ where every $\sigma_i$
      is a best response to $\bs{s}$ is a \emph{Nash equilibrium}.
    \end{definition}
    \begin{example}
      In a game of matching pennies, $T$ us a best response for Player One to
      $HH$, and $H$ is a best response for Player Two. The game has no
      equilibria in pure strategies.
    \end{example}
    
    The fact that games as simple as matching pennies may fail to have a
    pure strategy equilibrium casts doubt on its suitability as a solution
    concept. Fortunately, a natural extension of the framework rectifies the
    matter.
        
    \begin{definition}
      A \emph{mixed strategy} for player $i$ in a finite strategic game
      is a probability distribution over $S_i$.
      
      The utility player $i$ obtains from a profile of mixed strategies $S$
      is $\sum p(S')u_i(S')$, where $p(S')$ is the probability assigned
      to the pure profile $S'$ by the mixed strategies in $S$.
    \end{definition}
    
    It is in this context that Nash proved his seminal result:
    
    \begin{theorem}[\cite{Nash1951}]
      Every finite strategic game has an equilibrium in mixed strategies.
    \end{theorem}
    \begin{example}
      The unique equilibrium of matching pennies involves both players
      randomising over their sets of strategies by assigning a weight of
      $0.5$ to both $H$ and $T$. In this equilibrium both players
      attain a utility of $0.5$.
    \end{example}
    
    Since every game has an equilibrium, the algorithmic question of asking whether an equilibrium exists
    is not relevant. This motivates decision problems based on qualified notions of equilibria,
    such as the one that concerns us in this paper:
    


 \begin{quote}
$\exists$\GN: Given a Boolean game $G$ and a vector $\bs{v}\in[0,1]^n$, does $G$ have an equilibrium $\bs{s}$ such that $u_i(\bs{s})\geq\bs{v}[i]$ for each player $i$?
\end{quote}
    
    It is natural to also consider a problem closely related to the dual:
    
   
     \begin{quote}
  $\forall$\GN: Given a Boolean game $G$ and a vector $\bs{v}\in[0,1]^n$, does  every equilibrium of $G$, $\bs{s}$, satisfy $u_i(\bs{s})\geq\bs{v}[i]$ for each player $i$?
    \end{quote}
      
  \section{Main Result}
    Our reduction will be from the following NEXP-complete problem:

     \begin{quote}
  {\sc NEXPTM}: Given a non-deterministic Turing machine $M$, an integer in binary $K$ and a string $w$,
  does $M$ accept $w$ in at most $K$ steps?
    \end{quote}
  
    \begin{proposition}
      {\sc NEXPTM} is NEXP-complete.
    \end{proposition}
    \begin{proof}
      For membership in NEXP, we need only simulate the computation of $M$ on $w$ for
      $K$ steps. Each step can be simulated in non-deterministic polynomial time, and the
      number of steps is exponential in $|K|$.
    
      For hardness, let $N$ be a non-deterministic Turing machine with an exponential time clock $f$.
      Let $M$ be a Turing machine with an identical transition relation to $N$, but with 
      no internal clock. Clearly, $N$ accepts $w$ if and only if $M$ accepts $w$ in at most
      $f(w)$ steps. That is, $(M,f(w),w)$ is a positive instance of {\sc NEXPTM}. Moreover, the
      triple $(M,f(w),w)$ is polynomial in the size of $N$ and $w$: $|M|\leq |N|$, $|w|=|w|$ 
      and as $f(w)\in O(2^{p(|w|)})$, when written in binary it is of size at most $p(|w|)$. 
      This gives us the desired reduction.
    \end{proof}
  
    We can now prove the hardness of $\exists$\GN\. For questions of NEXP-membership,
    see the discussion below.
  
    \begin{theorem}\label{thm:NEXPhard}
      $\exists$\GN\ for Boolean games is NEXP-hard.
    \end{theorem}
    \begin{proof}
      We will give a reduction from {\sc NEXPTM}. Given a triple $(M,K,w)$ we shall 
      construct, in polynomial time, a Boolean game $G$ and a utility vector 
      $\bs{v}$, such that $G$ has an equilibrium where player $i$'s utility is at 
      least $\bs{v}[i]$ if and only if $M$ accepts $w$ in $K$ steps or less.

      For convenience, we augment $M$ with a ``do nothing" transition: if $M$ is at an 
      accepting state, then we allow it to move to the next computation step without
      moving the head, changing state, or writing anything to the tape. It is clear 
      that augmenting $M$ in such a fashion does not change the language accepted by $M$,
      but it ensures that the machine state is defined at all computation steps; we
      do not need to worry about the case where the machine accepts in under $K$ steps,
      as if it does, it will still accept at step $K$.
      
      Let $k=|K|$, and $q$ be the number of states of $M$.
      
      A computation history of $M$ on $w$ could be seen as a $K\times K$ table, or for
      simplicity $2^k\times 2^k$, padding as needed. Each
      row contains the tape contents and head position at a certain computation step. The
      number of bits needed to specify an entry of this table is $2k$.
      
      A way to visualise the proof is that in our game $G$, which consists of six
      players, Player One is equipped with variables that allow him to describe a
      single entry of this table. Player Four plays a partial matching pennies game
      against Player One, thereby forcing Player One to play a mixed strategy
      randomising over all entries of the table, and thus specifying an entire
      computation history with his mixed strategy. Player Two then verifies that
      the mixed strategy provided by Player One contains a consistent description
      of the head location at each computation step, and Player Three checks that
      every two consecutive steps are linked by exactly one transition rule. Players
      Five and Six play matching pennies with Players Two and Three to force them to
      randomise across all table entries.

      To this end, let:
      \begin{align*}
      \Phi_1&=\{\Zero_1,\One_1,\Head_1,\Left_1,\Right_1\}\\
      &\cup\{\Time_1^i\}_{1\leq i\leq k}\cup\{\Tape_1^i\}_{1\leq i\leq k}\cup\{\State_1^i\}_{1\leq i\leq q}.\\
      \end{align*}
      The intended meaning of $\Time_1^i$ (respectively $\Tape_1^i$)
      is the value of the $i$th most significant bit of the integer denoting the index of the
      computation step (respectively tape cell) in question, given the standard convention of interpreting 
      ``true" as 1 and ``false" as 0.
      A truth assignment by Player One can therefore be read as: at the computation step specified
      by $\Time_1^1,\dots,\Time_1^k$ the tape cell specified by $\Tape_1^1,\dots,\Tape_1^k$ contains
      0 if $\Zero_1$, 1 if $\One_1$ and is blank if neither. The machine head is hovering over the 
      cell in question if $\Head_1$, and is located to 
      the left or right
      of that cell respectively if $\Left_1$ or $\Right_1$. If the head is over the cell in question, 
      the machine is in state $i$ if $\State_1^i$ (if the head is not over the cell, $\State_1^i$ is
      a junk variable that has no meaning).
      
      Player One's goal formula is a conjunction of four subformulae:
      $$\gamma_1=\Init \wedge \Final\wedge \Cons_1\wedge \neg\gamma_4.$$
      
      Intuitively, $\Init$ means that if the player plays the first computation step, their
      description of the machine must agree with the initial configuration of $M$ on $w$.
      $\Final$ means that if the player plays the last ($K$th) computation step, the machine
      must be in an accepting state. $\Cons_1$ states the description of the machine must be internally consistent.
      The final conjunct is to force the player to randomise across all computation steps and
      tape cells, to which we will return later.
      
      $\Init$ requires that at time zero, the configuration of the machine is faithfully 
      represented by Player One's play. This takes the form of an implication where the antecedent
      states that we are at step zero:
      $$\Init=(\bigwedge_{1\leq i\leq k}\neg \Time_1^i)\ria \Consequent.$$
      The consequent itself is a conjunction of three further subformulae,
      for the head, the state and the tape.
      $$\Consequent=\InitHead\wedge \InitState\wedge \InitTape.$$
      
      The head requirement states that the head is at the leftmost cell. That is, at cell zero 
      $\Head_1$ is true, and
      at every other cell $\Left_1$ is true:
      \begin{align*}
	\InitHead&=\Big((\bigwedge_{1\leq i\leq k}\neg \Tape_1^i)\ria \Head_1\Big)\\ 
	&\wedge\Big(\neg(\bigwedge_{1\leq i\leq k}\neg \Tape_1^i)\ria \Left_1\Big).
      \end{align*}
      
      The state requirement is simply $M$'s initial state:
      $$\InitState=\State_1^{\mathit{initial}}.$$
      
      The tape requirement is a conjunction of $|w|+1$ implications. The first $|w|$ implications state that 
      if the tape cell chosen is within the first $|w|$ cells, then its contents must agree with
      $w$. If we use $i$ as shorthand for the conjunction of tape variables expressing $i$, and $w[i]$
      for $\Zero_1$ or $\One_1$ depending on the $i$th bit of $w$, this has the following form:
      $$\InitTape=\bigwedge_{0\leq i<|w|} (i\ria w[i])\wedge \BlankCells.$$
      
      Note that this formula is linear in $|w|$, so the construction so far
      was polynomial.
      
      The last formula 
      in $\InitTape$ states that all other cells are blank.
      $$\BlankCells= \neg(\bigvee_{0\leq i<|w|} i)\ria (\neg \Zero_1\wedge\neg \One_1).$$
      
      $\Final$ states that at computation step $K$, the machine accepts. If we use $K$ as shorthand
      for the appropriate conjunction of time variables, we get the following implication:
      $$\Final=K\ria \State_1^{\mathit{accepting}}.$$
      
      $\Cons_1$ requires that the player's description of a given computation step and cell is internally
      consistent. This means the cell cannot have both 0 and 1 on it, the head must be either over the cell
      or to one direction and the machine must be in exactly one state. It is worth noting that this says nothing
      about whether Player One's description of different steps and cells are consistent with each other: this
      is the task of Players Two and Three.

      For $\Cons_1$, we introduce a generalised XOR symbol, which we denote $\OneOf$, 
      with the interpretation that $\OneOf(\varphi_1,\dots,\varphi_n)$ is true if and only if
      exactly one $\varphi_i$ is. Such a symbol could be replaced by a propositional logic formula
      polynomial in the size of $\varphi_1,\dots,\varphi_n$ - simply take the disjunction of all $n$ admissible
      possibilities. This gives us the desired formula:
      \begin{align*}
	\Cons_1&=\neg(\Zero_1\wedge \One_1)\\
	&\wedge\OneOf(\Head_1,\Left_1,\Right_1)\\
	&\wedge\OneOf(\overline{\State_1^i}).
      \end{align*}
      By $\overline{\State_1^i}$ we mean $\State_1^1,\dots,\State_1^q$.
      
      To finish the description of $\gamma_1$, we turn to Player Four. Player Four
      is playing a partial matching pennies game with Player One over the time and tape
      variables. We thus equip her with the following:
      \begin{align*}
	\Phi_4&=\{\Time_4^i\}_{1\leq i\leq k}\cup\{\Tape_4^i\}_{1\leq i\leq k}.
      \end{align*}
      The objective is to guess the same computation step and cell index as player
      one:
      \begin{align*}
      \gamma_4&=\Big(\bigwedge_{1\leq i\leq k}(\Time_1^i\lra \Time_4^i)\Big)\\
      &\wedge
      \Big(\bigwedge_{1\leq i\leq k}(\Tape_1^i\lra \Tape_4^i)\Big).
      \end{align*}
      
      Player Two's purpose is to verify the consistency of Player One's description
      of the head. This involves verifying that at a given computation step the
      $\Head_1$ variable is true in exactly one cell, $\Left_1$ is true in every cell
      to the right and $\Right_1$ is true in every cell to the left. She controls
      the following variables:
      \begin{align*}
	\Phi_2&=\\
	&\{\Head_2,\sHead_2,\Left_2,\sLeft_2,\Right_2,\sRight_2\}\\
	&\cup\{\Tape_2^i\}_{1\leq i\leq k}\cup\{\sTape_2^i\}_{1\leq i\leq k}\cup\{\Time_2^i\}_{1\leq i\leq k}.
      \end{align*}
      
      The lowercase ``$s$" can be read as ``successor". The intended meaning of these variables is that
      $\Tape_2^1,\dots,\Tape_2^k$ name a cell and $\sTape_2^1,\dots,\sTape_2^k$ the cell directly to the
      right of it. The other variables state the location of the head in relation to these two cells
      at the computation step specified by the time variables.
      
      Player Two's goal formula is a conjunction of four subformulae:
      $$\gamma_2=\MatchOne_2\wedge \Cons_2\wedge \SUCC_2 \wedge \neg\gamma_5.$$
      
      Intuitively, $\MatchOne_2$ states that Player Two ought to play the same head configuration
      as dictated by Player One. $\Cons_2$ requires that this configuration be internally consistent.
      $\SUCC_2$ is to ensure that the two cells chosen are indeed consecutive.
      
      Before we state $\MatchOne_2$ we ought to first ask what we mean by saying that players one and two
      play the same head configuration. As in any given (pure) strategy profile, either player will be 
      describing a single computation step and at most two cells; if it turns out that they are speaking
      about different step/cell configurations we should not be concerned about whatever claims they
      make. Only in the instance where they happen to refer to the same step/cell should we expect
      accord. Since Player Two is referring to two cells in any play, we require that if either of the
      cells she references coincides with that referenced by Player One, they must agree.
      
      The desired formula is thus of the following form:
      \begin{align*}
	\MatchOne_2&=\\
	&\AgreeTime\ria\Big((\AgreeCell\ria \AgreeHead)\\
	&\wedge(\sAgreeCell\ria \sAgreeHead)\Big).
      \end{align*}
      
      The subformulae are as follows:
      \begin{align*}
	\AgreeTime&=\bigwedge_{1\leq i\leq k}(\Time_1^i\lra \Time_2^i).\\
	\AgreeCell&=\bigwedge_{1\leq i\leq k}(\Tape_1^i\lra \Tape_2^i).\\
	\sAgreeCell&=\bigwedge_{1\leq i\leq k}(\Tape_1^i\lra \sTape_2^i).\\
	\AgreeHead&=(\Head_1\lra \Head_2)\\
	&\wedge (\Left_1\lra \Left_2)\\
	&\wedge (\Right_1\lra \Right_2).\\
	\sAgreeHead&=(\Head_1\lra \sHead_2)\\
	&\wedge (\Left_1\lra \sLeft_2)\\
	&\wedge (\Right_1\lra \sRight_2).
      \end{align*}
      
      Internal consistency amounts simply to the conjunction of the valid combinations of
      claims about the head:
      \begin{align*}
	\Cons_2&=(\Right_2\wedge \sRight_2)\vee(\Right_2\wedge \sHead_2)\\
	&\vee
	(\Head_2\wedge \sLeft_2)\vee(\Left_2\wedge \sLeft_2)
      \end{align*}
      
      $\Succ_2$ states that the two tape locations are, in fact, consecutive. We will
      prove a lemma to show that this is concisely expressible in propositional
      logic.
      

      \begin{lemma}
	Let $\SUCC(p_1,\dots,p_n;q_1,\dots,q_n)$ be a formula that is true
	if and only if the binary integer encoded by $q_1,\dots, q_n$
	is the successor of the binary integer encoded by $p_1,\dots,p_n$. As a convention,
	$2^n-1$ has no successor.
	
	$\SUCC(p_1,\dots,p_n;q_1,\dots,q_n)$ can be replaced by a 
	propositional formula of size polynomial in $p_1,\dots,p_n$ and
	$q_1,\dots,q_n$.
      \end{lemma}
      \begin{proof}
	We take advantage of the fact that to increment a binary integer we only need to modify
	the rightmost consecutive block of 1s, and there are only $n$ such possible blocks.
	
	Since we have a boundary condition to consider, we require that the first integer is not
	$2^n-1$:
	$$\SUCC(p_1,\dots,p_n;q_1,\dots,q_n)=
	\neg(\bigwedge_{1\leq i\leq n} p_i)\wedge \Succ'.$$
	$\Succ'$ is then:
	\begin{align*}
	  &\Big(\neg p_1\ria\big(q_1\wedge\bigwedge_{i=2}^n(p_i\lra q_i)\big)\Big)\\
	  \wedge&\Big((p_1\wedge\neg p_2)\ria\big(\neg q_1\wedge q_2\wedge\bigwedge_{i=3}^n
	  (p_i\lra q_i)\big)\Big)\\
	  \wedge&\Big((p_1\wedge p_2\wedge\neg p_3)\ria\big(\neg q_1\wedge\neg q_2\wedge q_3
	  \wedge\bigwedge_{i=4}^n(p_i\lra q_i)\big)\Big)\\
	  &\vdots\\
	  \wedge&\Big((\neg p_n\wedge\bigwedge_{i=1}^{n-1} p_i)\ria\big((\bigwedge_{i=1}^{n-1} \neg q_i)\wedge q_{i+1}\big)\Big).\\
	\end{align*}
	This is quadratic in the number of variables, giving us the desired result.
      \end{proof}
      
      $\Succ_2$ can then be stated simply:
      $$\Succ_2=\SUCC(\overline{\Tape_2^i};\overline{\sTape_2^i}).$$
      
      Finally, Player Five is trying to guess Player Two's choice of cell and computation step.
      \begin{align*}
	\Phi_5&=\{\Time_5^i\}_{1\leq i\leq k}\cup\{\Tape_5^i\}_{1\leq i\leq k}.\\
      \gamma_5&=\bigwedge_{i=1}^k(\Tape_2^i\lra \Tape_5^i)\wedge
      \bigwedge_{i=1}^k(\Time_2^i\lra \Time_5^i).
    \end{align*}
      
      Player Three's purpose is to verify that the tape contents in successive computation steps
      respect the transition rules of $M$. To do this he specifies a total of six cells and two
      computation steps: consecutive triples in consecutive steps. Then he verifies that the tape
      contents, head position and machine state are in agreement with some rule of $M$.
      \begin{align*}
	\Phi_3=\{&\pHead_3, \Head_3,\sHead_3,\npHead_3, \nHead_3,\\ 
	&n\sHead_3, \pZero_3, \Zero_3, \sZero_3, \npZero_3,\\ 
	&\nZero_3,\nsZero_3,\pOne_3,\One_3,\sOne_3,\\
	&\npOne_3,\nOne_3,\nsOne_3\}\\
	\cup\{&p\State_3^i\}_{1\leq i\leq q}\cup\{\State_3^i\}_{1\leq i\leq q}\\
	\cup\{&s\State_3^i\}_{1\leq i\leq q}\cup\{np\State_3^i\}_{1\leq i\leq q}\\
	\cup\{&n\State_3^i\}_{1\leq i\leq q}\cup\{ns\State_3^i\}_{1\leq i\leq q}\\
	\cup\{&p\Tape_3^i\}_{1\leq i\leq k}\cup\{\Tape_3^i\}_{1\leq i\leq k}\\
	\cup\{&\sTape_3^i\}_{1\leq i\leq k}\cup\npTape_3^i\}_{1\leq i\leq k}\\
	\cup\{&\nTape_3^i\}_{1\leq i\leq k}\cup\{\nsTape_3^i\}_{1\leq i\leq k}\\
	\cup\{&\Time_3^i\}_{1\leq i\leq k}\cup\{\nTime_3^i\}_{1\leq i\leq k}.
      \end{align*}
      The ``$p$" can be read as ``predecessor", referring to the cell to the left, and ``$n$" as ``next computation step".
      The intended meaning is simply the state and tape contents in each of the six cells, as well as whether the
      head is over that cell.
      
      Player Three's goal formula is a conjunction of five subformulae:
      $$\gamma_3=\MatchOne_3\wedge \Triple\wedge \Succ_3\wedge \Rules \wedge \neg\gamma_6.$$
      
      $\MatchOne_3$ states that if any of the step/cell pairs named by Player Three coincide
      with the one named by Player One, Player Three must agree with Player One. $\Triple$ requires
      that the three cells named in either computation step should be a consecutive triple, and the triple
      at either step must be the same. $\Succ_3$ requires
      the two computation steps named to be consecutive. $\Rules$ is to verify that the configuration
      thus described is consistent with a rule of $M$.
      
      $\MatchOne_3$ is a conjunction of a total of six statements, depending on which step/cell
      pair coincides with that played by Player One. We will only give one such statement below,
      in the case that Player One named the same step as $\Time_3^1,\dots,\Time_3^k$ and the same
      cell as $p\Tape_3^1,\dots,p\Tape_3^k$. The other five statements are obtained in the obvious
      manner.
      \begin{align*}
	\Big(&\bigwedge_{i=1}^k(\Time_1^i\lra \Time_3^i)\wedge
	\bigwedge_{i=1}^k(\Tape_1^i\lra p\Tape_3^i)\Big)\ria\\
	\Big((&\Zero_1\lra p\Zero_3)\wedge (\One_1
	\lra p\One_3)\\
	&\wedge (\Head_1\lra p\Head_3)\wedge\bigwedge_{i=1}^q(\State_1^i\lra p\State_3^i)\Big).
      \end{align*}
      
      $\Triple$ states that the tape cells selected are consecutive triples, and that the same triple
      is chosen in both steps. It is worth noting
      that given our previous definition of successor, if Player Three is to satisfy this conjunct
      then the middle cell cannot be 0 or $2^k-1$.
      \begin{align*}
	\Triple&=\SUCC(\overline{\pTape_3^i};\overline{\Tape_3^i})\\
	&\wedge\SUCC(\overline{\Tape_3^i};\overline{\sTape_3^i})\\
	&\wedge\SUCC(\overline{\npTape_3^i};\overline{\nTape_3^i})\\
	&\wedge\SUCC(\overline{\nTape_3^i};\overline{\nsTape_3^i}\\
	&\wedge\bigwedge_{1\leq i\leq k}(\Tape_3^i\lra \nTape_3^i).
      \end{align*}
      
      $\Succ_3$ requires that the computation steps be consecutive:
      $$\Succ_3=\SUCC(\overline{\Time_3^i};\overline{\nTime_3^i}).$$
      
      $\Rules$ is a conjunction of four formulae: three of the formulae are conjunctions containing an implication for each $(r,s)\in Q\times\{0,1,\bot\}$,
      representing the machine's behaviour if it reads $s$ in state $r$ and the head is over the left, centre or right cell respectively. The fourth
      term is $\NoHead$, to handle the case where the head is not over any cell in the triple:
      $$\Rules=\Left\wedge \Centre\wedge \Right\wedge\NoHead.$$
      We will examine $\Left$ and $\NoHead$, understanding that $\Centre$ and $\Right$ are handled in similar fashion.
      \begin{align*}
      \Left&=\Big(\bigwedge_{(r,s)\in Q\times\{0,1,\bot\}}\big((\pState_3^r\wedge s)\ria\\
      &\OneOf(\overline{\Rule[(r,s)\ria(r',s',D)]})\big)\Big).
      \end{align*}
      The $s$ in the antecedent is meant to be replaced by $\pZero_3$, $\pOne_3$, or $\neg(\pZero_3\vee \pOne_3)$ as
      appropriate. The intuition of the $\Rules$ term is that should the machine read $s$ in state $r$ it should pick exactly
      one of the rules available to it, and if the head is not present then the tape contents should not change.
      
      The subformula to deal with a specific rule can be broken up as follows:
      $$\Rule[(r,s)\ria(r',s',D)]=L\wedge B.$$
      $L$ describes the behaviour of the machine if the left cell is not the leftmost cell on the tape, 
      $B$ deals with the boundary
      case where it is.
      
      We will give an example of how $\Rule[(q_3,0)\ria(q_4,1,L)]$ would be handled. All rules except
      ``do nothing" can be handled similarly, and ``do nothing" would merely assert that if the machine
      reads an accepting state, then nothing changes.
      
      The $L$ part triggers if the head is over the leftmost cell in the triple, and the leftmost cell is not cell 0.
      It then ensures that in the next computation step the leftmost cell contains 1 and the other cells are unchanged.
      Since the head leaves the monitored triples we need no terms to account for it.
      \begin{align*}
	L=&\Big(\neg(\bigwedge_{1\leq i\leq k}\neg \pTape_3^i)\wedge \pHead_3\Big)\ria\\
	&\Big(\npOne_3\wedge
	(\Zero_3\lra \nZero_3)\\
	&\wedge(\sZero_3\lra \nsZero_3)\wedge(\One_3\lra \nOne_3)\\
	&\wedge(\sOne_3\lra \nsOne_3)\Big).
      \end{align*}
      In the boundary case the head is over the leftmost cell of the tape, so when it attempts to move
      left it instead stands still.
      \begin{align*}
	B=&\Big((\bigwedge_{1\leq i\leq k}\neg \pTape_3^i)\wedge \pHead_3\Big)\ria\\
	&\Big(\npState_3^4\wedge \npOne_3\wedge \npHead_3\\
	&(\Zero_3\lra \nZero_3)
	\wedge(\sZero_3\lra \nsZero_3)\\
	&\wedge(\One_3\lra \nOne_3)\wedge(\sOne_3\lra \nsOne_3)\Big).
      \end{align*}
      
      Finally, the $\NoHead$ term asserts in the absence of a head the tape contents do not change.
      \begin{align*}
      \NoHead&=(\neg \pHead_3\wedge\neg \Head_3\wedge\neg \sHead_3)\ria\\
      &\Big((\pOne_3\lra \npOne_3)\wedge(\nZero_3\lra \npZero_3)\\
      &\wedge(\One_3\lra \nOne_3)\wedge(\Zero_3\lra \nZero_3)\\
      &\wedge(\sOne_3\lra \nsOne_3)\wedge(\sZero_3\lra \nsZero_3)
      \Big).
      \end{align*}
      This brings us to the last player, who is trying to guess the first step and central 
      cell chosen by Player Three:
      \begin{align*}
	\Phi_6&=\{\Time_6^i\}_{1\leq i\leq k}\cup\{\Tape_6^i\}_{1\leq i\leq k}.\\
      \gamma_6&=\bigwedge_{i=1}^k(\Tape_3^i\lra \Tape_6^i)\wedge
      \bigwedge_{i=1}^k(\Time_3^i\lra \Time_6^i).
  \end{align*}

      The construction so far has been polynomial. We now claim that $M$ having an accepting 
      run on $w$ in at most $K$ steps is equivalent
      to the constructed game having a Nash equilibrium where Players One, Two and Three
      have the following guaranteed payoffs:
      \begin{align*}
	\bs{v}[1]&=\frac{2^{2k}-1}{2^{2k}}.\\
	\bs{v}[2]&=\frac{2^{k}(2^k-1)-1}{2^{k}(2^k-1)}.\\
	\bs{v}[3]&=\frac{(2^k-2)(2^k-1)-1}{(2^k-2)(2^k-1)}.
      \end{align*}
      
      First, suppose $M$ has an accepting run on $w$ in at most $K$ steps. 
      Consider the profile where Player One randomises over all step/cell combinations
      with equal weight, and at each step/cell combination plays his variables in
      accordance to the accepting run. Player Four also randomises over all 
      step/cell combinations with equal weight. Player Two randomises over all computation
      steps and the first $2^k-1$ cells. Her other variables she plays in accordance to the
      run. Player Five likewise randomises over all steps and the first $2^k-1$ cells.
      Player Three randomises over the first $2^k-1$ steps and the $2^k-2$ cells between
      the first and last. His other variables he plays in accordance to the run. Player
      Six randomises over the same $2^k-1$ steps and the $2^k-2$ cells.
      
      In such a profile, Players One, Two and Three will satisfy their goals unless their
      step/cell combination is guessed by their opponent. Given our setup, this will happen with
      probabilities $1/2^{2k}$, $1/2^{k}(2^k-1)$ and $1/(2^k-2)(2^k-1)$ respectively, giving us
      the payoffs $\bs{v}[1]$, $\bs{v}[2]$ and $\bs{v}[3]$. It remains to see that this profile is in equilibrium.
      
      Let us first consider Players Four through Six. Any pure strategy by Player Four is a
      step/cell pair, and hence, given the play of Player One, has a $1/2^{2k}$ chance of
      satisfying $\gamma_4$. Player Four is thus indifferent between the current situation
      and any deviation. For Player Five any pure strategy using the first $2^k-1$ cells will
      have a $1/2^{k}(2^k-1)$ chance of satisfying $\gamma_5$, and any other pure strategy 0.
      Player Five thus likewise has no incentive to deviate. In the same fashion, any
      pure strategy for Player Six will satisfy $\gamma_6$ with probability $1/(2^k-2)(2^k-1)$ or
      0, so she is also indifferent.
      
      In the case of Player One, observe that no matter what pure strategy he picks, there is a
      $1/2^{2k}$ chance of Player Four guessing the cell/step component and thus making $\gamma_1$
      false. It follows that any such strategy will yield at most a $\bs{v}[1]$ chance of satisfying $\gamma_1$.
      For Player Two, if she picks a pure strategy using the first $2^k-1$ cells there will likewise
      be a $1/2^{k}(2^k-1)$ chance of her step/cell combination being guessed. If she picks a pure
      strategy using the last cell, she will be unable to satisfy the $\Succ_2$ component of $\gamma_2$,
      yielding a utility of 0. For Player Three, any pure strategy using the $2^k-1$ steps and the
      $2^k-2$ cells randomised over by six will have a $1/(2^k-2)(2^k-1)$ chance of being guessed,
      and any other choice of pure strategy will violate either $\Triple$ or $\Succ_3$. This establishes 
  that the described profile is in equilibrium.
      
      Next, suppose that no accepting run exists. We claim that in any equilibrium Player One
      will still obtain a utility of $\bs{v}[1]$, but either Player Two or three will be unable to secure a 
    payoff of $\bs{v}[2]$, $\bs{v}[3]$. For the first part, note that for any choice of step/cell by Player One, 
    the remaining variables
      can be set to satisfy $\Init$, $\Final$ and $\Cons_1$ unilaterally. It is sufficient to simply
      respect the initial configuration of the machine at step zero, play an accepting state at
      step $K$, and any internally consistent description elsewhere. Any strategy that does not
      satisfy $\Init$, $\Final$ and $\Cons_1$ is thus dominated and can be excluded from consideration.
      All that remains is the choice of cell/step and it is easy to see that the only equilibrium play
      would involve giving every pair equal weight.
      
      Player One's play will thus describe a sequence of $2^k$ configurations of $M$, with the initial
      configuration at step zero and an accepting state at step $K$. However, as $M$ has no accepting
      run on $w$ in $K$ steps, this sequence cannot represent a valid computation and a violation
      must occur somewhere.
      
      If this violation involves the assertion of the presence of more than one head or 
      the $\Left_1$, $\Right_1$ variables
      incorrectly specifying the location of the head, we claim that Player Two cannot obtain a
      utility of $\bs{v}[2]$.
      
      Observe that in this case there must exist two consecutive cells at some time step where Player
      One plays one of the following combinations:
      \begin{center}
      \begin{tabular}{l | l}
	Cell $i$& Cell $i+1$\\
	\hline
	$\Left_1$&$\Right_1$\\
	$\Left_1$&$\Head_1$\\
	$\Head_1$&$\Head_1$\\
	$\Right_1$&$\Left_1$\\
	$\Head_1$&$\Right_1$\\
      \end{tabular}
      \end{center}
      In this case, should Player Two play a strategy involving cell $i$, since she
      is committed to playing a legal head assignment she will have to disagree
      with Player One on either cell $i$ or cell $i+1$. This means she will suffer a $1/2^{2k}$
      chance of having $\MatchOne$ falsified if Player One plays the cell in question. As there is
      still at least a $1/2^{k}(2^k-1)$ chance of having the cell/step combination guessed by Player Five,
      this means the maximum utility Player Two can obtain in this case is $\bs{v}[2]-1/2^{2k}+1/2^{2k}2^{k}(2^k-1)$.
      (The last term is to avoid double counting the case where both Player One and Player Five name the same cell/step combination.)
      
      Of course, Player Two may opt in this case not to play any strategies involving cell $i$.
      This will however mean that she is randomising over at most $(2^k-2)$ cells, and
      Player Five will randomise accordingly, meaning the highest utility she can obtain is
      $\frac{2^{k}(2^k-2)-1}{2^{k}(2^k-2)}$.

      Suppose now that Player One does not make such a violation. The remaining possibilities
      for an incorrect run are:

      \begin{enumerate}
  \item
    The head make an illegal transition.
  \item
    The tape contents undergo an illegal change.
  \item
    The state undergoes an illegal change.
      \end{enumerate}


      Let us deal with case 1. Suppose between step $t$ and $t+1$ the head, which is at cell $i$ at $t$, 
      performs an illegal transition.
      This could mean moving more than one cell in a direction, moving off the edge of the tape, staying still
      in a non-accepting state or moving one cell left or right without a justifying transition rule. 
      Observe that neither of these possibilities is consistent with the $\Rules$ requirement. As such,
      should Player Three pick step $t$ and cell $i$, he will have to disagree with Player One on
      the movement of the head, thereby running a risk of falsifying his formula should Player One
      play $t$ and $i$. This will prevent Player Three from obtaining $v_3$ utility for the same
      reasoning as with Player Two.
      
      In case 2, there would exist steps $t$ and $t+1$, and a cell $i$ the contents of which would change
      without a justifying rule. This, too, would violate $\Rules$. For case 3, we note that by the machine
      state we mean the state variable that occurs in the same cell as the head: the value of the other state
      variables is of no account. As such, $\Rules$ again would be violated as it requires the correct state to
      be propagated to cell hosting the head. This completes the proof.
    \end{proof}
    
    We can adapt this proof to show that $\forall$\GN\ is coNEXP-hard. Note that this does not follow
    immediately: $\forall$\GN\ is not simply the complement of $\exists$\GN. Letting $\bs{s}$ range over equilibrium profiles,
    $\forall$\GN\ is the question whether:
    $$\forall\bs{s} \st \forall i \st u_i(\bs{s})\geq\bs{v}[i]$$
    the complement of $\forall$\GN\ is then:
    $$\exists\bs{s} \st \exists i \st u_i(\bs{s})<\bs{v}[i].$$
    To show that $\forall$\GN\ is coNEXP-hard we need only show that the latter problem is NEXP-hard.
    
    \begin{corollary}
      $\forall$\GN\ is coNEXP-hard.
    \end{corollary}
    \begin{proof}
      We argue that the proof of Theorem~\ref{thm:NEXPhard} can be adapted to show this. Note that the utilities of Players
      One, Four, Five and Six did not play a r\^ole in the proof. Those of Four, Five and Six were omitted entirely,
      whereas Player One has been seen to achieve $\bs{v}[1]$ utility in every equilibrium. What remains are Two and
      Three, and we will argue that those players could be collapsed into a single player.
      
      Introduce a new player into the game constructed in the proof of Theorem~\ref{thm:NEXPhard}, Player Seven, with $\gamma_7=\gamma_2\wedge\gamma_3$
      and $\Phi_7=\emptyset$. We argue that the Turing machine $M$ accepts $w$ in at most $K$ steps if and only if there exists an $\bs{s}$ for which:
      
      $$u_7(\bs{s})\geq 1-\frac{(2^k-2)(2^k-1)+2^k(2^k-1)-1}{2^k(2^k-1)(2^k-2)(2^k-1)}.$$
      This can be seen by replicating the argument in the proof: in the presence of an accepting run, the only way Player Seven
      can lose utility is if Player Five or Six guesses the same cell/step, which happens with probabilities $\frac{1}{2^k(2^k-1)}$ and $\frac{1}{(2^k-2)(2^k-1)}$
      respectively. Adding a term for double counting and simplifying yields the quantity above.
      
      We have thus shown that the following question is NEXP-hard:
      $$\exists\bs{s} \st \exists i \st u_i(\bs{s})\geq\bs{v}[i].$$
      
      For the next step, add Player Eight with $\gamma_8=\neg\gamma_7$ and $\Phi_8=\emptyset$. As $u_8=1-u_7$ the following
      question is NEXP-hard as well, letting $\bs{v}[8]=1-\bs{v}[7]$:
      $$\exists\bs{s} \st \exists i \st u_i(\bs{s})\leq\bs{v}[i].$$
      It remains to show that the inequality can be made strict. 
      
      First, observe that we can increase Player Seven's score, and hence decrease
      Player Eight's, by an arbitrarily small
      $\epsilon$ of a certain form: let $\gamma_7'=\gamma_7\vee\mathit{Pennies}$ where $\mathit{Pennies}$ is a matching pennies game over a new set of variables
      $\Phi_7'$ against some new player. This will give Player Seven $1/2^{|\Phi_7'|}$ additional utility, minus a double counting term. 
      
      All that remains is to show that we can identify a ``sufficiently small" $\epsilon$. By this we mean an $\epsilon$ satisfying
      the following:
      $$\exists\bs{s} \st u_8(\bs{s})-\epsilon<\bs{v}[8]\iff\exists\bs{s} \st u_8(\bs{s})\leq\bs{v}[8].$$
      To see that this is possible, recall that if $M$ does not accept $w$ in $K$ steps, then Player One necessarily specifies
      an incorrect computation history of the machine. As we have seen in the proof of Theorem~\ref{thm:NEXPhard}, such a violation
      decreases the maximum attainable score of Player Two or Three by a fixed amount. It is thus possible to calculate the maximum attainable utility of
      Player Seven in the presence of such a violation, which will give us the bounds within which $\epsilon$ may reside.
      
      This completes the proof.
    \end{proof}
  
    \subsection{Discussion}
    
      The preceding proof raises two related questions. To begin with, one may ask whether
      six players are necessary. The answer is no: the reader may convince themselves
      that one may reduce the number to three in a straightforward fashion by collapsing
      Players Two and Three, and Four, Five and Six onto each other, in a similar fashion to the proof of the corollary. 
      We used six players
      to simplify the exposition of the proof. Whether it is further
      possible to reduce the number to two is a different matter.
      
      Second: whether there is a membership result to go with the hardness. Strictly speaking,
      there is not. As there exist games where every equilibrium requires irrational weights
      on the strategies chosen (\cite{Nash1951}; \cite{Bilo2012})
      we cannot rely on the intuitive approach of guessing a strategy profile and checking whether
      it is in equilibrium.
      
      One way this problem is addressed in the literature is to restrict attention to two player
      games, where a rational equilibrium \emph{is} guaranteed to exist. This brings us back
      to the first question. The second way is to consider the notion of an $\epsilon$-equilibrium:
      a profile of strategies where no player can gain more than $\epsilon$ utility by deviating.
      This problem, $\epsilon$-$\exists$\GN, clearly does belong to NEXP and the reader can convince themselves
      that by inserting a sufficiently small $\epsilon$ into the proof above we can establish that it
      is NEXP-complete.
  
  \section{Conclusion}
  
    We have shown that the problem of determining whether a Boolean game has a Nash equilibrium which
    guarantees each player a certain payoff is NEXP-hard. This is the first complexity result about
    mixed equilibria in the Boolean games framework, and demonstrates that in this instance Boolean games
    are as difficult as the more general class of Boolean circuit games.
    
    The complexity of many other natural problems remains open, most significantly that of {\sc Nash}:
    the task of computing a mixed equilibrium. However, given the difficulty in obtaining this result
    for normal form games \cite{Daskalakis2006} one could posit that it is unlikely that this can
    be achieved with the current tools of complexity theory. It would be interesting to see whether
    there is an exponential time analogue of PPAD that could lead to a solution to this problem.
  
  \section{Acknowledgements}
  
    Egor Ianovski is supported by a scholarship, and Luke Ong is partially supported by a grant, from the Oxford-Man Institute of Quantitative Finance.
   

\bibliographystyle{alpha}
\bibliography{references}
\end{document}